\newif\ifarxiv
\arxivtrue

\ifarxiv

\documentclass{article}
\usepackage{amsthm}
\usepackage{amssymb}
\usepackage{amsmath}
\usepackage{hyperref}
\usepackage{cleveref}
\usepackage{thm-restate}
\usepackage{caption}
\usepackage{subcaption}
\usepackage[a4paper,total={140mm,222mm}]{geometry}

\newtheorem{theorem}{Theorem}
\newtheorem{lemma}[theorem]{Lemma}
\newtheorem{corollary}[theorem]{Corollary}
\theoremstyle{definition}

\title{Strategic Facility Location in Euclidean Spaces} %

\else

\documentclass[a4paper,UKenglish,cleveref, autoref, thm-restate]{lipics-v2021}

\title{Strategic Facility Location in Euclidean Spaces} %

\author{Kim Thắng {Nguyễn}}{Université Grenoble Alpes, Grenoble INP, CNRS, INRIA, LIG, Grenoble, France}{kim-thang.nguyen@univ-grenoble-alpes.fr}{}{}
\author{Lucas Perotin}{Université Grenoble Alpes, Grenoble INP, CNRS, INRIA, LIG, Grenoble, France}{lucas.perotin@univ-grenoble-alpes.fr}{}{}
\author{Bertrand Simon}{Université Grenoble Alpes, CNRS, Grenoble INP, INRIA, LIG, Grenoble, France}{bertrand.simon@cnrs.fr}{}{}

\authorrunning{K.T. Nguyễn, L. Perotin and B. Simon } %

\Copyright{Kim Thắng Nguyễn, Lucas Perotin and Bertrand Simon} %


\keywords{Facility Location, Strategy-proofness, Euclidean spaces} %

\category{} %

\relatedversion{} %

\acknowledgements{This work is supported by the ANR project Predictions ANR-23-CE48-0010 and the MIAI Chaire Frugal Artificial Intelligence ANR-23-IACL-0006.}%

\fi

\usepackage[numbers]{natbib}
\usepackage{amsthm}

\usepackage[utf8]{inputenc} %
\usepackage[T5]{fontenc} %

\newtheorem{hypothesis}{Hypothesis}

\usepackage{xspace}
\usepackage{tablefootnote}
\usepackage{pgfplots}
\usetikzlibrary{arrows.meta,calc,decorations.markings}
\pgfplotsset{compat=1.18}

\newcommand{\RR}{\ensuremath{\mathbb R}\xspace}
\newcommand{\RRp}{\ensuremath{\mathbb R^{\phantom{2}}}\xspace}
\newcommand{\OPT}{\ensuremath{\mathit{OPT}}\xspace}
\newcommand{\cost}{\ensuremath{\mathit{cost}}\xspace}

\usepackage{amsmath}
\usepackage{dsfont}
\usepackage{bm}
\usepackage{todonotes}

\newcommand{\vect}[1]{\ensuremath{\bm{#1}}}

\newcommand{\E}{\ensuremath{\mathbb{E}}\xspace}
\newcommand{\M}{\ensuremath{\mathcal{M}}\xspace}
\newcommand{\II}{\ensuremath{\mathcal{I}}\xspace}
\newcommand{\OO}{\ensuremath{\mathcal{O}}\xspace}

\ifarxiv

\date{}

\author{
    Kim Thắng {Nguyễn}, %
	Lucas Perotin, %
    Bertrand Simon %
}
\date{
\small{
Université Grenoble Alpes, Grenoble INP, CNRS, INRIA, LIG, Grenoble, France\\
    kim-thang.nguyen@univ-grenoble-alpes.fr, lucas.perotin@univ-grenoble-alpes.fr, bertrand.simon@cnrs.fr
}}

\fi

\begin{document}

\maketitle

\begin{abstract}
The strategic facility location problem is defined as follows: $n$ agents report
their location in a metric space, and the objective is to design a
\emph{mechanism} deciding the (possibly randomized) location of a facility such
that agents have no incentive to lie about their position. We focus on the
egalitarian cost, which means that the goal of the mechanism is to minimize the
expected maximal facility-agent distance. Meanwhile, mechanisms must be \emph{truthful} (or
\emph{strategyproof}): no agent may decrease their expected distance to the
facility via lying on their location.

Designing truthful mechanisms minimizing the approximation ratio is a
well-studied problem, and the optimal solution is known for the real line. We
focus in this paper on higher dimension Euclidean spaces, for which gaps remain
between the best known lower and upper bounds. We first show that, maybe
counter-intuitively, the problem is easier for two agents on the plane rather
than on the line: the mechanism can exploit the additional dimension to prevent
more efficiently agent lies. Based on this intuition, we devise lower bounds for
$\mathbb R^d$ asymptotically matching the best known approximation factor of $2$
for large $d$. We also provide novel mechanism ideas, improving over the best
known algorithms on the plane, and when the agents belong to $\mathbb R^d$ but
the facility may use an additional dimension.
\end{abstract}

\newpage

\section{Introduction}

Mechanism design is a well-established research area in which payments are used to incentivize participants to truthfully report their preferences, with the goal of designing mechanisms that achieve certain social objectives. Such mechanisms are called \emph{strategyproof} or \emph{truthful}. While payments play a central role in the design of strategyproof mechanisms, in many settings, such as medical domains and public services, monetary transfers are not an appropriate means of incentivization. Designing truthful mechanisms without payments in general settings is impossible under unrestricted preferences, as formalized by the celebrated Gibbard--Satterthwaite impossibility theorem. However, in specific settings, strategyproofness can be achieved by relaxing the optimality requirements of the social objective. Procaccia and Tennenholtz~\cite{procaccia2013approximate} pursued this direction and initiated the study of mechanisms without payments for the facility location problem, which has since become a central benchmark in approximate mechanism design without money.

In the facility location problem, there are $n$ agents located in a metric space, and the goal is to select a location at which to open a facility. The rule used to select the facility location is public knowledge and may be either deterministic or randomized. Each agent has a private location and reports a location to the designer. The facility location is then chosen based on the reported locations. Agents are self-interested, and each aims to minimize their \emph{individual cost}, defined as the (expected) distance from their true location to the facility. Thus, an agent may attempt to manipulate the outcome of the facility location rule by strategically misreporting their position. The designer’s goal is to construct a mechanism under which every agent has an incentive to report truthfully, while also minimizing a social cost objective. Two standard aggregate objectives are the \emph{utilitarian cost}, defined as the sum of all agents’ costs, and the \emph{egalitarian cost}, defined as the maximum cost over all agents.

The facility location problem in mechanism design without money was studied extensively about a decade ago. For general metrics, tight bounds were established for both the utilitarian~\cite{thang2010group} and egalitarian~\cite{alon2010strategyproof} objectives. These bounds are essentially achieved by dictatorship mechanisms and their randomized counterparts. Recently, interest in the problem has been revived, motivated by beyond-worst-case analysis and by the design of algorithms and mechanisms with predictions. Specifically, the lower-bound constructions in~\cite{alon2010strategyproof,thang2010group} are based on graph metrics, whereas practical applications of facility location often arise in Euclidean spaces. This motivates the study of facility location mechanisms in the Euclidean space $\mathbb{R}^{d}$.  
In particular, one aims for the design of new mechanisms beyond the classic dictatorship framework and improved approximation bounds. 

We focus in this paper on the egalitarian objective, thus minimizing the
distance between the output facility and the farthest agent. In this case, it is
known that the best deterministic algorithm returns any agent location and
achieves an approximation ratio of 2 (applying the triangle inequality to bound
the distance between any pair of agents by the sum of their distances to the
optimal facility). We therefore consider randomized algorithms, and aim at
minimizing the expected maximal distance between the facility and any agent.

When restricted to the one-dimensional case, in which all agents and the
facility must lie on the real line, the optimal randomized mechanism for the
egalitarian cost is known~\cite{procaccia2013approximate} and achieves an
approximation ratio of $1.5$, which is tight even for two agents. In two
dimensions, this problem remains open. Furthermore, the best known lower bound
is weaker and equals $\sqrt{5/4}$~\cite{balkanski2024randomized}, valid also for
two agents. This is counterintuitive, as the problem seems harder for the
mechanism. Restricted to two agents, it was conjectured
in~\cite{balkanski2024randomized} that the facility should lie on the line
connecting the agents, which would imply the validity of the $1.5$ lower bound.

In this paper, we refute this conjecture: there exists a
$\sqrt{2}$-approximation algorithm for two agents on the plane, which is
\emph{optimal}. This surprising result suggests that allowing the facility to be opened in a larger space 
than the agents' feasible space can be beneficial, even if it
is obvious that the optimal facility location lies on the original space.
This is the consequence of the additional possible locations allowing to penalize
more efficiently algorithm lies. Most randomized mechanisms use some center
(e.g., centroid) of the agents reported positions to determine the facility
location. An agent intuitively has an interest to lie in the direction opposite
to this center to bias the mechanism towards its location. Allowing the facility
to penalize the spread of the agent positions using the additional possible locations
therefore allows for better mechanisms.

We therefore consider the setting in which the space \II in which the agents may
reside and report their position is not necessarily equal to the space \OO in
which the mechanism may open a facility. This intuition led to the design of a
lower bound of $2-\varepsilon$ for the randomized facility location problem in
$\mathbb R^d$ for large $d$ and $n$, thus asymptotically matching the upper
bound. We further improve lower and upper bounds on the plane. We also
generalize the idea of using an additional dimension to design a better
deterministic mechanism when $\II=\mathbb R^d$ and $\OO=\mathbb R^{d+1}$.
Finally, we show that in the specific setting $\II=\mathbb R^2$ and $\OO=\mathbb
R^3$, a different generalization leads to a $\sqrt{3}$-approximation, which can
be improved when restricted to three agents.

\subsection{Summary of contributions}

As discussed above, our first result consists in the design and optimality proof
of a $\sqrt{2}$-approximation mechanism in the plane. The mechanism places the
facility at the two remaining vertices of the square whose other two opposite
vertices are the agents' reported positions. The result stays valid when the
agents are restricted to the line, for any $n\geq 2$, or when there are two
agents in the plane; see Section~\ref{sec:two-agents}.
Note that throughout this paper, we focus on truthfulness \emph{in expectation}
regarding randomized algorithms, as detailed in the next section. The proof of
the lower bound uses new ideas based on a potential function over the plane
that measures the \emph{truthfulness} and the
\emph{efficiency} of a facility located at each point. The expected cost of a
randomized mechanism is therefore lower bounded by the potential minimum.

\begin{theorem}
\label{th:optimal-sqrt-two}
The optimal approximation ratio achievable by a truthful mechanism is
$\sqrt{2}$ in either of the following settings: $\II=\mathbb R$ and
$\OO=\mathbb R^2$ for any $n\geq 2$, or $\II=\OO=\mathbb R^2$ for $n=2$.
In the first setting, the optimal mechanism is deterministic. Lower bounds
assume mechanisms invariant under translations, rotations, and scalings.
\end{theorem}

Based on the intuitions gained by this surprising result, we devise a recursive
lower bound construction for $\mathbb R^d$ in which a fraction of $n$ agents
are iteratively moved onto the vertices of $q$-simplices in new dimensions. This
results in a lower bound for any randomized mechanism of $2-\varepsilon$ for large $d$ and
$n$, asymptotically matching the best known upper bound of $2-1/n$; see Section~\ref{sec:lbrd}.
This construction can be tuned for $\mathbb R^2$ to improve the lower bound in this case.

\begin{theorem}
Any truthful randomized mechanism on $\mathbb R^d$ and $n$ agents has an approximation ratio at least $2 - O\left( \frac{\log d}{d} + \frac{1}{n^{1/d}} \right)$, and at least $1.6054 - O\left(\frac1n\right)$ for $d=2$.
\end{theorem}

The classical randomized mechanism with approximation ratio $2-1/n$ is known as the
\emph{Centroid}~\cite{tang2020characterization}: with probability $1/2$, output
the centroid (isobarycenter) of the agents' reported positions, and otherwise
output a randomly drawn agent position. In~\cite{balkanski2024randomized}, the
authors consider a variant of the problem where the mechanism has further access
to unreliable predictions. They show that using a different center, which was
never used to the best of our knowledge for the egalitarian objective, allows
for a better mechanism on the plane. We show that this center can actually be leveraged to improve
over the classical algorithm in the plane, thus also improving over the $2-1/n$ approximation factor stated as
an open problem in~\cite{balkanski2024randomized}, see Section~\ref{sec:R2R2UB}. A direct generalisation also leads to an improvement on $\mathbb R^3$.

\begin{theorem}
\label{th:R2UB}
There exists a randomized truthful mechanism on $\mathbb R^2$ with an approximation factor at most $2-\frac{3-\sqrt 2}{n+2}$, which can be extended to a $2-\frac{3-\sqrt{3}}{n+2}$-approximation on $\mathbb R ^3$.
\end{theorem}

We then focus on the low-dimensional setting $\II=\mathbb R^2$ and $\OO=\mathbb R
^3$. This mechanism generalizes the $\sqrt{2}$-approximation mechanism by
combining an additional dimension with the center used in Theorem~\ref{th:R2UB},
and yields a better competitive ratio. Unfortunately, this mechanism is not
truthful in the more general setting $\II=\mathbb R^3$, and the idea leads to
worse approximation ratios for larger dimensions, see Section~\ref{sec:r2-r3}.

\begin{theorem}
There exists a deterministic truthful mechanism for $\II=\mathbb R^2$ and $\OO=\mathbb R^{3}$ with an approximation factor equal to $\sqrt 3$.
\end{theorem}

Using a different generalization of the $\sqrt{2}$ algorithm exploiting an additional
dimension for the facility, we show that this approach can be used for any
dimension, by placing the facility \emph{above} the centroid, lifted by a
distance equal to the standard deviation of the agent positions, see Section~\ref{sec:RdRdp}.

\begin{theorem}
There exists a deterministic truthful mechanism for $\II=\mathbb R^d$ and $\OO=\mathbb R^{d+1}$ with an approximation factor equal to $2\sqrt{\frac{n-1}{n}}$.
\end{theorem}

\subsection{Concurrent works}

During the writing of this manuscript, three concurrent works have been
published on Arxiv~\cite{gomes2026improved,hastings2026maximum,barak2026ex}. The
earlier one~\cite{gomes2026improved} designed the same $\sqrt{2}$-approximation
mechanism for two agents in the plane, with a different proof idea for its
optimality, a weaker lower bound for $\mathbb R^d$, and a
$\sqrt{d+1}$-approximation mechanism for $\II=\mathbb R^d$ and $\OO=\mathbb
R^{d+1}$ matching the mechanism we propose for $d=2$, but unusable for larger
$d$, contrarily to our $2\sqrt{\frac{n-1}{n}}$ mechanism. We also reused the
notation for $\II$ and $\OO$ they introduce to improve the consistency over
public manuscripts. The second one~\cite{hastings2026maximum} shows a lower
bound construction similar to ours for $\mathbb R^d$, though our construction
for $\mathbb R^2$ implies a stronger lower bound. The latter one~\cite{barak2026ex} shows another
lower bound construction for $\mathbb R^d$, as well as the same lower bound as
ours for $\mathbb R ^2$.

\subsection{Related Work}

There has been some early work on characterizing strategyproof facility location mechanisms without
payments on lines~\cite{moulin1980strategy} and on circles~\cite{SchummerVohra:01}. 
However, \citet{procaccia2013approximate} initiated in the study of 
approximating the optimum social cost under the constraint of (group) strategyproofness. 
The authors considered facility location problems on line-metric spaces 
and in particular gave tight approximation ratio bounds of 2 and of 3/2 for deterministic 
and randomized mechanisms respectively under egalitarian objective. 
In general metric spaces, tight approximation ratios of 2 
have been estalished in \cite{alon2010strategyproof} and \cite{alon2009strategyproof,thang2010group} 
for egalitarian and utilitarian objectives. 
Subsequently, many studies have focuses on this problem under various scenarios 
and generalization on the number of facilities 
(few facilities \cite{LuSunWang:10,chan2026randomized,ma2026breaking,jia2026product,aziz2026anchoring} 
and many facilities \cite{escoffier2011strategy}), on strong notion of group-strategyproofness 
(for example, \cite{procaccia2013approximate,thang2010group}), etc. 
We refer to the survey~\cite{HauAris21:Mechanism-Design} for further discussions.

The interest in this problem has been revived recently using the recent
\emph{learning-augmented} framework where the mechanism has further access to
untrusted predictions, both in the deterministic~\cite{agrawal2022learning} and
randomized~\cite{balkanski2024randomized} settings. Besides, the Facility Location problem in 
Euclidean spaces is attractive since it is closely related to practical scenarios. Moreover, 
the consideration of Euclidean spaces opens opportunities for the design of new mechanisms 
beyond the dictatorship mechanisms. 
This direction has highlighted relevant open directions, even without the predictions framework,
and have motivated this paper.

\subsection{Preliminaries}

We consider the facility location problem in which a mechanism has access to the reported location profile of $n$ agents $\vect{x} = (x_{1}, \ldots, x_{n})$ belonging to a Euclidean space $\II$.
A randomized mechanism $\M$ then
outputs a random variable representing the facility location $F = \M(\vect{x})$ belonging to a Euclidean space $\OO\supseteq\II$. Unless explicitly stated otherwise, we assume $\OO=\II$. Its objective is to minimize the expected egalitarian cost, i.e., $\cost(\M,\vect{x}) = \E[\max_i d(F,x_i)]$, where $d$ is the Euclidean distance.

A mechanism is \emph{strategyproof}, or \emph{truthful in expectation}, if,
for every location profile $\vect{x}$, every agent $i$, and every alternative
report $x'_i\in\II$, we have
\[
\E[d(x_i, \M(\vect{x}))]
\leq \E[d(x_i, \M(x_1,\dots,x_{i-1},x'_i,x_{i+1},\dots,x_n))].
\]

This ensures that no agent has an incentive to lie on its actual location, as this would not decrease the expected distance from the facility to its actual location.
Note that stronger requirements exist such as \emph{universal truthfulness} requiring randomized mechanism to draw the facility location using truthful deterministic mechanisms, which are much more restrictive in our context.

The optimal cost assuming agents report their true locations is equal to the radius of the smallest enclosing hypersphere of $\vect{x}$: $\OPT(\vect{x}) = \min_{c\in\II}\max_i d(x_i,c)$. A mechanism \M is an $\alpha$-approximation if, for any $\vect{x}\in\II^n$, we have $\cost(\M,\vect{x})\leq \alpha \cdot \OPT(\vect{x})$.

\section{Two agents on the Euclidean plane}

\label{sec:two-agents}
\subsection{The perpendicular square lottery mechanism}
\label{sec:two-agent-plane}

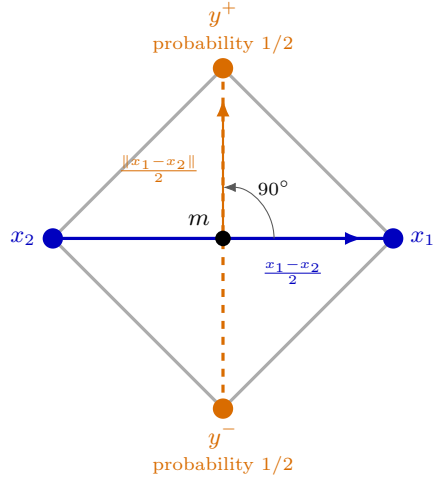
\begin{figure}[htbp]
\centering
\begin{tikzpicture}[>=Latex,every node/.style={font=\small},scale=2.25]
  \coordinate (Xone) at (1,0);
  \coordinate (Xtwo) at (-1,0);
  \coordinate (M) at (0,0);
  \coordinate (Yplus) at (0,1);
  \coordinate (Yminus) at (0,-1);

  \draw[gray!65,very thick] (Xone)--(Yplus)--(Xtwo)--(Yminus)--cycle;
  \draw[blue!75!black,very thick] (Xtwo)--(Xone);
  \draw[orange!85!black,very thick,dashed] (Yminus)--(Yplus);

  \draw[->,blue!75!black,thick] (M)--(.82,0)
    node[midway,below=4pt,font=\scriptsize]
    {$\frac{x_1-x_2}{2}$};
  \draw[->,orange!85!black,thick] (M)--(0,.82)
    node[midway,left=5pt,font=\scriptsize]
    {$\frac{\|x_1-x_2\|}{2}$};
  \draw[->,black!65,thin] (.30,0)
    arc[start angle=0,end angle=90,radius=.30];
  \node[font=\scriptsize,fill=white,inner sep=.5pt] at (.30,.30)
    {$90^\circ$};

  \fill[blue!75!black] (Xone) circle (1.7pt)
    node[right=3pt] {$x_1$};
  \fill[blue!75!black] (Xtwo) circle (1.7pt)
    node[left=3pt] {$x_2$};
  \fill[orange!85!black] (Yplus) circle (1.7pt)
    node[above=2pt,align=center]
    {$y^+$\\[-1pt]{\scriptsize probability $1/2$}};
  \fill[orange!85!black] (Yminus) circle (1.7pt)
    node[below=2pt,align=center]
    {$y^-$\\[-1pt]{\scriptsize probability $1/2$}};
  \fill (M) circle (1.3pt) node[above left=2pt] {$m$};
\end{tikzpicture}
\caption{Definition of the perpendicular square lottery. The reports
$x_1,x_2$ form one diagonal of the square. The two possible facility
locations $y^+$ and $y^-$ form the other diagonal, perpendicular to
$[x_1,x_2]$, and are each selected with probability $1/2$.}
\label{fig:perpendicular-square-mechanism}
\end{figure}

We consider the standard setting with input and output spaces
$\mathcal I=\mathcal O=\mathbb R^2$ and two agents. Given reports
$x_1,x_2\in\mathbb R^2$, let
\[
  m=\frac{x_1+x_2}{2}.
\]
Let $y^+$ and $y^-$ be the two points on the line perpendicular to
the segment $[x_1,x_2]$ through $m$, one on each side of
$[x_1,x_2]$, at distance $\|x_1-x_2\|/2$ from $m$.
The \emph{perpendicular square lottery}
$\M_{\perp}$ returns $y^+$ or $y^-$, each with probability $1/2$.
Thus $x_1,x_2,y^+,y^-$ are the four vertices of a square, and the reported
segment and the output segment are its two diagonals. The construction is
shown in \Cref{fig:perpendicular-square-mechanism}.

\begin{theorem}
\label{thm:square-upper}
For two agents in $\mathbb R^2$, the perpendicular square lottery is
truthful in expectation and has exact approximation ratio $\sqrt2$
for the expected ex-post maximum cost.
\end{theorem}

\begin{proof}
We first prove truthfulness. If $x_1=x_2$, the mecanism outputs the position of the agents and they have no incentive to lie. Otherwise, without loss of generality, assume that
$x_1=0$ and $x_2=1$. Consider a deviation $x_2'$ of agent~2, and let
$y^+$ and $y^-$ be the two possible facility locations. Thus
$x_1,x_2',y^+,y^-$ are the four vertices of a square, with $x_1$ and
$x_2'$ as opposite vertices. We denote coordinates in $\mathbb R^2$ as complex numbers in $\mathbb C$.

\begin{figure}[t]
\centering
\begin{tikzpicture}[>=Latex, every node/.style={font=\small}, scale=3.0]
  \coordinate (O) at (0,0);
  \coordinate (One) at (1,0);
  \coordinate (I) at (0,1);
  \coordinate (Q) at (1.5,-0.5);
  \coordinate (Yplus) at (1,0.5);
  \coordinate (Yminus) at (0.5,-1);

  \path[use as bounding box] (-0.18,-1.16) rectangle (1.74,1.17);
  \draw[->,gray!55] (-0.14,0)--(1.70,0);
  \draw[->,gray!55] (0,-1.10)--(0,1.12);

  \tikzset{
    square side/.style={blue!75!black,very thick,
      postaction={decorate},decoration={markings,
      mark=at position .5 with {\draw[blue!75!black,thick]
      (-2.6pt,-2.1pt)--(0,0)--(-2.6pt,2.1pt);}}},
    equal axis/.style={black,thick,
      postaction={decorate},decoration={markings,
      mark=at position .58 with {\draw[black,thick]
      (-1.6pt,-3pt)--(-1.6pt,3pt);
      \draw[black,thick](1.6pt,-3pt)--(1.6pt,3pt);}}},
    equal segment/.style={orange!85!black,very thick,
      postaction={decorate},decoration={markings,
      mark=at position .5 with {\draw[orange!85!black,thick]
      (-2pt,-3pt)--(2pt,3pt);}}}
  }

  \draw[square side] (O)--(Yplus)--(Q)--(Yminus)--cycle;
  \draw[equal axis] (O)--(One);
  \draw[equal axis] (O)--(I);
  \draw[equal segment] (Yminus)--(One);
  \draw[equal segment] (Yplus)--(I);
  \draw[black!65,very thick] (One)--(Yplus);
  \draw[black,very thick] (One)--(I)
    node[midway,above right,font=\scriptsize] {$\sqrt2$};

  \draw[blue!75!black,thick] (O) ++(-63.435:0.20)
    arc (-63.435:26.565:0.20);
  \node[blue!75!black,font=\scriptsize] at ($(O)+(-21.435:0.26)$)
    {$90^\circ$};
  \draw[black,thick] (O) ++(-63.435:0.33) arc (-63.435:0:0.33);
  \draw[black,thick] (O) ++(26.565:0.33) arc (26.565:90:0.33);
  \node[font=\scriptsize] at ($(O)+(-31.7:0.4)$) {$\theta$};
  \node[font=\scriptsize] at ($(O)+(58.3:0.4)$) {$\theta$};

  \fill (O) circle (1.2pt) node[below left] {$0$};
  \fill (One) circle (1.4pt) node[below right] {$1$};
  \fill (I) circle (1.4pt) node[above left] {$i$};
  \fill[green!45!black] (Q) circle (1.4pt) node[above right] {$x'_2$};
  \fill[blue!80!black] (Yplus) circle (1.5pt)
    node[above right] {$y^+$};
  \fill[red!75!black] (Yminus) circle (1.5pt)
    node[below] {$y^-$};

  \node[draw=black!45,rounded corners,fill=white,inner sep=2pt,
    anchor=west,font=\scriptsize] (eq) at (1.08,0.93)
    {$d(y^+,i)=d(y^-,1)$};
  \draw[-{Latex[length=1.5mm]},orange!85!black,thin]
    (eq.west) to[bend left=10] ($(Yplus)!0.5!(I)$);
\end{tikzpicture}
\caption{A normalized off-line deviation.  The reports are $0$ and
$x_2'$, and the lottery returns the other two vertices $y^+,y^-$ of
the blue square.  Rotation by $90^\circ$ sends
$(0,y^-,1)$ to $(0,y^+,i)$, so the two orange segments have equal
length.  The agent's expected cost is half the broken path
$1-y^+-i$, whose length is at least $|1-i|=\sqrt2$.}
\label{fig:square-deviation}
\end{figure}
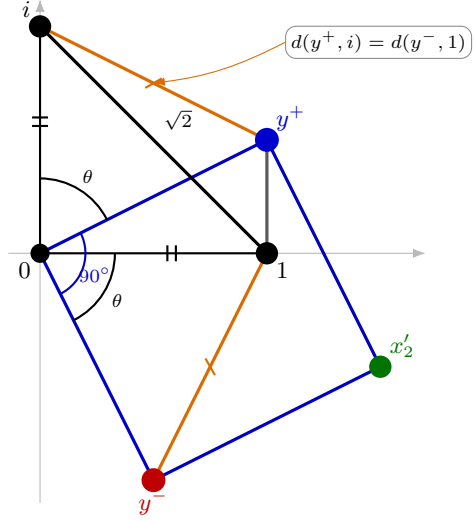

Consider the two triangles $(0,y^-,1)$ and $(0,y^+,i)$ illustrated in
\Cref{fig:square-deviation}. Since $0,y^+,x_2',y^-$ form a square, we have
$    d(0,y^-)=d(0,y^+)$ 
and the angle
$
    \widehat{y^-0y^+}=90^\circ=\widehat{10i}.
$
As $d(0,1)=d(0,i)$,
it follows that
the triangle $(y^+,0,i)$ can be obtained from $(y^-,0,1)$ by a rotation of $90^\circ$ centered on $0$.
Therefore,
$
    d(y^-,1)=d(y^+,i).
$

Hence, the expected cost of agent~2 after the deviation is
\[
\frac{d(1,y^+)+d(1,y^-)}{2}=
\frac{d(1,y^+)+d(y^+,i)}{2}\geq \frac{d(1,i)}{2}
=\frac{1}{\sqrt{2}},
\]
where the inequality follows from the triangle inequality.

When agent~2 reports truthfully, both possible facility locations are at
distance $1/\sqrt{2}$ from its true location. Therefore, no deviation can
decrease its expected cost. The argument for agent~1 is symmetric, proving
truthfulness in expectation.

It remains to compute the approximation ratio. Under our normalization,
the two agents are at distance $1$ from each other. Each output and each
agent are at distance $1/2$ from the midpoint of $[x_1,x_2]$, in
perpendicular directions. Hence, the realized maximum cost is therefore always $1/\sqrt{2}$, whereas
the optimum is the midpoint and has cost $1/2$. The approximation ratio
is thus exactly $\sqrt{2}$.
\end{proof}

\subsection{Tightness and the line-to-plane setting}
\label{sec:two-agent-lower}

In this section, we consider the following natural assumption.

\begin{hypothesis}
\label{hyp:invariance}
The mechanism is invariant under
translations, rotations, and positive scalings.
\end{hypothesis}

Under \Cref{hyp:invariance}, we show that the $\sqrt
2$-approximation designed in the previous section for 2 agents in $\mathbb R^2$
is  the best possible using a new technique. We assume the 2 agent true
positions fixed and define a potential function $V$ on the plane, where the
function argument corresponds to a facility location output by a
mechanism. $V$ is defined as the sum of two terms, one representing the cost of
the facility (the maximum agent-facility distance), and the other being positive
if a given lie by an agent (moving away from the other agent) is beneficial, and negative otherwise. We know that
the expected value of the second term is nonpositive by truthfulness of the
mechanism. Therefore, a lower bound on the potential function implies a lower
bound on the expected cost. We then note that restricting $\II$ to $\mathbb R$ does not impact the validity of this lower bound.

\begin{theorem}
\label{thm:square-lower}
Every finite-ratio, truthful-in-expectation two-agent mechanism in
$\mathbb R^2$ that is invariant under translations, rotations, and
positive scalings has approximation ratio at least $\sqrt2$.
\end{theorem}

\begin{proof}
We may first anonymize the mechanism.  Replace $\M$ by the equal mixture of
$\M(x_1,x_2)$ and $\M(x_2,x_1)$.  In the second term, a deviation by the
first reported agent is a deviation by the second agent of the original
mechanism.  Hence both terms, and therefore their mixture, are
truthful.  This operation preserves \Cref{hyp:invariance} and cannot increase
the worst-case ratio. We identify locations in $\mathbb R ^2$ as complex numbers in $\mathbb C$ in 
order to simplify the notations.

By \Cref{hyp:invariance}, we may consider the normalized truthful
profile $(x_1,x_2)=(1,0)\in\mathbb C ^2$, and let
$F=\M(1,0)\in\mathbb C$ denote the random facility location.  The direct similarity
$y\mapsto1-y$ defined on $\mathbb C$ exchanges the two reports.  By
\Cref{hyp:invariance} and anonymity, we therefore
have
\begin{equation}
\label{eq:endpoint-symmetry}
  F\mathrel{\stackrel{d}{=}}1-F,
\end{equation}
where $\stackrel{d}{=}$ denotes equality in distribution.

Consider the agent truly located at $1\in\mathbb C$.  If it reports
$1+\varepsilon$, with $\varepsilon>0$, while the other report remains
$0$, \Cref{hyp:invariance} applied to the positive scaling
$y\mapsto(1+\varepsilon)y$ implies that the facility is distributed as
$(1+\varepsilon)F$.  For a realized output $z\in\mathbb C$, define
\begin{equation}
\label{eq:outward-cost-slope}
  \Delta(z):=\lim_{\varepsilon\downarrow0}
  \frac{|(1+\varepsilon)z-1|-|z-1|}{\varepsilon}.
\end{equation}
This is the first-order change in the agent's cost under an outward lie.  A
negative value means that the lie brings the facility closer; a positive
value means that it moves the facility away.  Direct differentiation gives
\begin{equation}
\label{eq:outward-cost-slope-formula}
  \Delta(z)=
  \begin{cases}
    \displaystyle
    \frac{\langle z,z-1\rangle}{|z-1|}
    =\frac{|z|^2+|1-z|^2-1}{2|1-z|},&z\ne1,\\[7pt]
    1,&z=1.
  \end{cases}
\end{equation}
The second line is especially simple: if the facility is already at the
agent's true position, then after the lie its location becomes
$1+\varepsilon$, so the agent's cost grows from $0$ to
$\varepsilon$.  Thus this output gives no local incentive to lie.

Truthfulness implies
\[
  \E[\bigl|(1+\varepsilon)F-1\bigr|]
  \ge \E[|F-1|].
\]
Dividing by $\varepsilon$ and letting $\varepsilon\downarrow0$ yields
\begin{equation}
\label{eq:expected-outward-slope}
  \E[\Delta(F)]\ge0.
\end{equation}

We now turn this incentive constraint into a pointwise potential.  Define the
symmetrized truthfulness correction
\begin{equation}
\label{eq:truthfulness-correction}
  T(z):=-(\Delta(1-z)+\Delta(z)).
\end{equation}
The two terms measure the outward cost speeds relative to the two endpoints.
By \Cref{eq:endpoint-symmetry,eq:expected-outward-slope},
\begin{equation}
\label{eq:expected-truthfulness-correction}
  \E[T(F)]
  =-\left(\E[\Delta(F)]+\E[\Delta(1-F)]\right)
  =-2\E[\Delta(F)]
  \le0.
\end{equation}

For a realized facility location $z\in\mathbb C$, define:
\[
  a=|z|,
  \qquad b=|1-z|,
  \qquad D(z)=\max\{a,b\}.
\]
Here $D(z)$ is exactly the realized maximum cost at the normalized profile.
For $z\in\mathbb C\setminus\{0,1\}$, \Cref{eq:outward-cost-slope-formula} gives
\begin{equation}
\label{eq:truthfulness-correction-formula}
  T(z)=-\frac{(a+b)(a^2+b^2-1)}{2ab}.
\end{equation}
We define the pointwise potential
\[
  V(z):=D(z)+\frac{T(z)}{4}.
\]

We claim that
\begin{equation}
\label{eq:geometric-potential-bound}
  V(z)\ge\frac1{\sqrt2}
  \qquad\text{for every }z\in\mathbb C.
\end{equation}

If the facility is placed at location $z$  close to the midpoint, then $D(z)$ is small but $T(z)$
is large, reflecting an incentive for an endpoint agent to report farther
outward. Far from the midpoint, this incentive disappears, but $D(z)$ is
large. The factor $1/4$ balances these two effects in the pointwise
inequality. Since truthfulness gives $\E[T(F)]/4\leq 0$, the bound
$\E[V(F)]\geq 1/\sqrt{2}$ implies $\E[D(F)]\geq 1/\sqrt{2}$. As the
optimum is $1/2$, the approximation ratio is at least $\sqrt{2}$.
See Figure~\ref{fig:potential2agents} for an illustration of these functions.

\begin{figure}[htbp]
    \centering
    \begin{subfigure}[b]{0.32\textwidth}
        \centering
        \includegraphics[width=\textwidth, page=1]{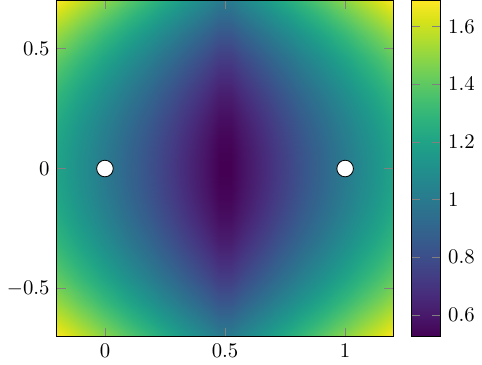}
        \caption{$D(z)$}
    \end{subfigure}
    \begin{subfigure}[b]{0.32\textwidth}
        \centering
        \includegraphics[width=\textwidth, page=2]{fig/plot_sqrt2_heatmap.pdf}
        \caption{$T(z)$, null on the circle except at $0,1$}
    \end{subfigure}
    \begin{subfigure}[b]{0.32\textwidth}
        \centering
        \includegraphics[width=\textwidth, page=3]{fig/plot_sqrt2_heatmap.pdf}
        \caption{$V(z)$, minimal on stars}
    \end{subfigure}
    \caption{Illustration of the potential function $V(z)$ and its two terms $D(z)$ and $T(z)$. The two points represent the agent locations.}
    \label{fig:potential2agents}
\end{figure}

We first prove \Cref{eq:geometric-potential-bound} away from the endpoints.  The three side lengths
$a,b,1$ satisfy $|a-b|\le1\le a+b$.  By symmetry assume $a\ge b$, so
$a\ge1/2$, $b\in[|1-a|,a]$, and $D(z)=a$.  Set
\[
  \Gamma(a,b)=\frac{(a+b)(a^2+b^2-1)}{4ab},
  \qquad
  V(z)=a-\frac12\Gamma(a,b).
\]
For fixed $a$,
\[
  \frac{\partial\Gamma}{\partial b}(a,b)
  =\frac{2b^3+ab^2+a-a^3}{4ab^2}.
\]
The numerator is increasing in $b>0$, because its derivative is
$6b^2+2ab>0$.  Hence $\Gamma(a,b)$ has no interior maximum on the
admissible interval, and it is enough to inspect the endpoints.  At $b=a$,
\[
  \Gamma(a,a)=a-\frac1{2a}.
\]
At $b=|1-a|$, its value is $-1/2$ when $1/2\le a<1$, and
$a-1/2$ when $a>1$; both are at most $a-1/(2a)$.  When $a=1$, the
same conclusion follows directly from $\Gamma(1,b)=b(1+b)/4$.  Therefore
\[
  V(z)
  \ge a-\frac12\left(a-\frac1{2a}\right)
  =\frac a2+\frac1{4a}
  \ge\frac1{\sqrt2},
\]
where the last inequality is arithmetic-geometric mean.

It remains only to check the endpoints, where we use the definition
\Cref{eq:truthfulness-correction} directly rather than the formula containing
$ab$ in the denominator.  We have $\Delta(1)=1$ and $\Delta(0)=0$, so
\[
  T(0)=T(1)=-1,
  \qquad
  V(0)=V(1)=1-\frac14=\frac34>\frac{\sqrt2}{2}.
\]
This expresses the intended tradeoff: placing the facility at an endpoint is
good for incentives because the colocated agent does not want to move its
report outward, but it is socially far from the other agent.  The potential
records both effects, and the distance term is large enough to give a strict
margin.  This completes the proof of \Cref{eq:geometric-potential-bound}.

Finally, take expectations in \Cref{eq:geometric-potential-bound}.  Using
\Cref{eq:expected-truthfulness-correction},
\[
  \E[D(F)]
  =\E[V(F)]-\frac14\E[T(F)]
  \ge\E[V(F)]
  \ge\frac1{\sqrt2}.
\]
The optimum for the profile $(1,0)$ is $1/2$.  The approximation ratio is
therefore at least $(1/\sqrt2)/(1/2)=\sqrt2$.
\end{proof}

Combining \Cref{thm:square-upper,thm:square-lower}, the optimal ratio within
this invariant class is exactly $\sqrt2$.  Equality in the pointwise
potential bound requires $a=b=1/\sqrt2$, which gives exactly the two square
vertices $(1+i)/2$ and $(1-i)/2$.

\paragraph*{Agents on the line with a facility in the plane.}
We now consider $\II=\mathbb R$ and $\OO=\mathbb R^2$ for any number
$n\geq 2$ of agents. We identify $\mathbb R$ with the horizontal axis of
$\mathbb R^2$. Given a report profile $\vect{x}=(x_1,\ldots,x_n)$, let

$$
  x_-=\min_i x_i,
  \qquad
  x_+=\max_i x_i,
  \qquad
  m=\frac{x_-+x_+}{2},
  \qquad
  a=\frac{x_+-x_-}{2}.
$$

The \emph{lifted interval mechanism} $\M_{\mathrm{int}}$ places the facility
at $(m,a)$. Equivalently, the facility lies vertically above the midpoint of
$[x_-,x_+]$, at height $(x_+-x_-)/2$.

\begin{theorem}
\label{thm:lifted-interval}
For $\II=\mathbb R$ and $\OO=\mathbb R^2$, the lifted interval mechanism is
deterministic and truthful for every $n\geq 2$, and its exact
approximation ratio is $\sqrt{2}$. This ratio is optimal under the invariance
assumptions of \Cref{thm:square-lower}.
\end{theorem}

\begin{proof}
Fix an agent with true location $x$ and fix the reports of all other agents.
Let $x_-$ and $x_+$ be, respectively, the smallest and largest of these
reports. If the agent reports $x'$, the new extremes are

$$
  x_-'=\min\{x_-,x'\},
  \qquad
  x_+'=\max\{x_+,x'\}.
$$

Writing

$$
  m'=\frac{x_-'+x_+'}{2},
  \qquad
  a'=\frac{x_+'-x_-'}{2},
$$

the squared cost of the agent is
\begin{align*}
\bigl|(x,0)-(m',a')\bigr|^2
&=(x-m')^2+(a')^2 \
=\frac{(x-x_-')^2+(x-x_+')^2}{2}.
\end{align*}
We show that this expression is minimized by $x'=x$. If
$x\in[x_-,x_+]$, every report inside $[x_-,x_+]$ leaves both extremes
unchanged, while a report outside this interval replaces one extreme by a
point farther from $x$. If $x<x_-$, truthful reporting gives
$x_-'=x$ and $x_+'=x_+$. For every report $x'$, we have
$x_+'\geq x_+$, and therefore

$$
  \frac{(x-x_-')^2+(x-x_+')^2}{2}
  \geq \frac{(x-x_+)^2}{2},
$$

which is precisely the squared truthful cost. The case $x>x_+$ is symmetric.
Since the square root is increasing, truthful reporting also minimizes the
agent's cost, proving truthfulness.

It remains to compute the approximation ratio. If $x_-=x_+$, the mechanism
is optimal. Otherwise, every agent lies in $[x_-,x_+]$, so its horizontal
distance from $m$ is at most $a$. Hence its cost is at most

$$
  \sqrt{a^2+a^2}=\sqrt{2}\,a,
$$

with equality for the agents at the two extremes. Any facility has maximum
cost at least half the distance $x_+-x_-=2a$ between these extremes, while
the midpoint $(m,0)$ attains maximum cost $a$. Thus $\OPT=a$, and the exact
ratio of $\M_{\mathrm{int}}$ is $\sqrt{2}$.

Finally, the lower-bound proof of \Cref{thm:square-lower} applies to the
present setting: its normalized profile has reports $0$ and $1$, and the
only deviation it uses is the outward report $1+\varepsilon$, which also lies
in $\mathbb R$. Therefore restricting the agents' reports to the line does
not affect that proof, and the ratio $\sqrt{2}$ is optimal under the same
invariance assumptions.
\end{proof}

\section{\texorpdfstring{Lower bounds for randomized strategy-proof mechanisms in $\mathbb{R}^d$}{Lower bounds for randomized strategy-proof mechanisms in Rd}}
\label{sec:lbrd}

In this section, we will show that an iterative construction of an instance in
$\mathbb R^d$ for $n$ agents leads to an asymptotically tight lower bound of
$2-\varepsilon$ for large $d$ and $n$, for truthful randomized mechanisms.

We start by showing a lower bound tuned to $\mathbb R^2$, using a simpler
construction in the same spirit but adapted to the number of dimensions
available.

\subsection{\texorpdfstring{A lower bound construction specifically tuned for $\mathbb{R}^{2}$}{A lower bound construction specifically tuned for {R}{2}}}

As discussed above, we first devise a construction in which multiple instances are
considered successively, to lower bound the approximation ratio of a truthful
mechanism in $\mathbb R^2$. The idea is to rely on the mechanism truthfulness to
transfer properties on expected facility-vertex distances from one instance to
the following one, and use such properties to devise stronger ones in the new instance.

\begin{theorem}
Every strategy-proof mechanism in $\mathbb{R}^{2}$ has an approximation ratio at least 
$1 + \frac{1+2\sqrt7}{6\sqrt3} - O\bigl( \frac1n \bigr) \approx 1.6054 - O(\frac 1n)$.
\end{theorem}
\begin{proof}
Consider an instance $I_0$ in which $n/2$ agents are located on $O=(0,0)$ and $n/2$ on $A=(1,0)$. There must exist one of these points for which the expected distance to the facility is at least $0.5$ by the triangle inequality and the linearity of expectation. Assume it is $O$. See Figure~\ref{fig:LB2} for an illustration.

\begin{figure}[htbp]
    \centering
    \begin{subfigure}[b]{0.32\textwidth}
        \centering
        \includegraphics[width=\textwidth, page=1]{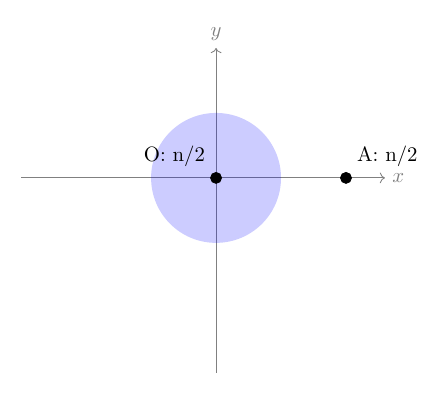}
        \caption{Instance $I_0$}
    \end{subfigure}
    \begin{subfigure}[b]{0.32\textwidth}
        \centering
        \includegraphics[width=\textwidth, page=2]{fig/plot_LB2.pdf}
        \caption{Instance $I_1$}
    \end{subfigure}
    \begin{subfigure}[b]{0.32\textwidth}
        \centering
        \includegraphics[width=\textwidth, page=3]{fig/plot_LB2.pdf}
        \caption{Instance $I_2$}
    \end{subfigure}
    \caption{Illustration of the lower bound construction. The number next to points and the dotted circle represent the number of agents at this location. The blue and red disks  represent the minimum expected distance from the facility to their center.}
    \label{fig:LB2}
\end{figure}

Consider the instance $I_1$ in which $n/2$ agents are located on $(1,0)$, and $n/4$ on each of $B=(-1/2,\frac{\sqrt{3}}{2})$ and $C=(-\frac12,-\frac{\sqrt{3}}{2})$. These vertices form an equilateral triangle centered on $O$ of side length $\sqrt{3}$ and circumradius $1$. As $I_1$ can be obtained from $I_0$ by $n/2$ consecutive lies of agents located at $(0,0)$, the expected distance from the facility to $O$ cannot decrease from $I_0$ to $I_1$: 
$$ \E_{F \sim \M(I_{1})} \left[ d(F,O)\right] \geq  \E_{F \sim \M(I_{0})} \left[ d(F,0)\right]\geq 0.5.$$ 

For a fixed facility location $F$, let $\rho=d(F,O)$. By
\cite[Theorem~1.2]{nikolov2011sum}, among all points at distance $\rho$
from the center of a regular triangle, the sum of the distances to its
three vertices is minimized when the point lies on a ray from the center
through one of the vertices. Hence,
\[
    \frac{1}{3}\sum_{v\in\{A,B,C\}} d(F,v)
    \geq
    \phi(\rho)
    :=
    \frac{|1-\rho|+2\sqrt{\rho^2+\rho+1}}{3},
\]
since in such an extremal position the three distances are
\[
    |1-\rho|,\qquad
    \sqrt{\rho^2+\rho+1},\qquad
    \sqrt{\rho^2+\rho+1}.
\]

The function $\phi$ is convex and nondecreasing on $\mathbb R_+$.
Therefore, by Jensen's inequality and
$\E_{F\sim\M(I_1)}[d(F,O)]\geq 1/2$,
\begin{align*}
\max_{v\in\{A,B,C\}}
\E_{F\sim\M(I_1)}[d(F,v)]
&\geq
\frac13\sum_{v\in\{A,B,C\}}
\E_{F\sim\M(I_1)}[d(F,v)]\\
&\geq
\E_{F\sim\M(I_1)}[\phi(d(F,O))]\\
&\geq
\phi\!\left(\E_{F\sim\M(I_1)}[d(F,O)]\right)\\
&\geq
\phi(1/2)
=
\frac{1+2\sqrt7}{6}.
\end{align*}

$$\max_{v\in\{A,B,C\}}  \E_{F \sim \M(I_{1})} \left[ d(F,v)\right] \geq \frac{1+2\sqrt{7}}{6}\approx1.049>1.$$

Let $v$ be the vertex achieving this maximal expected distance.

Now consider an instance $I_{2}$ obtained from instance $I_{1}$ by spreading $n/4$ agents 
from vertex $v$ uniformly onto the circle, centered at $v$ and of radius $\sqrt{3}$. Note that, by this uniform spreading, 
any agent on the circle has a distance of at most 
$O(1/n)$ from another neighbor-agent on the circle. The critical observation now is that for every location of the facility $F$, 
there exists an agent located at a point $p_{F}$ on the circle  ``almost'' on the opposite side of $F$ through the center $v$ of the circle. 
In other words, $p_{F}$ is of distance at most $O(1/n)$ to the intersection of the circle and the line $Fv$ farthest to $v$. 
Therefore, the expected cost of the mechanism in instance $I_{2}$ is
\begin{align*}
\E_{F \sim \M(I_{2})} \left[ \max_{p \in I_2} d(F, p) \right]
&\geq \E_{F \sim \M(I_{2})} \left[ d(F, p_{F}) \right] \\
&\geq \E_{F \sim \M(I_{2})} \left[ d(F, v) + d(v, p_{F}) - O\bigl( \frac 1n \bigr) \right] \\
&\geq \E_{F \sim \M(I_{1})} \left[ d(F, v) \right] + \sqrt{3}  - O\bigl( \frac 1n \bigr) \\
&\geq \frac{1+2\sqrt{7}}{6} + \sqrt{3}  - O\bigl( \frac1n \bigr) 
\end{align*}
The third inequality follows from $\E_{F \sim \M(I_{2})} \left[ d(F,v)\right] \geq  \E_{F \sim \M(I_{1})} \left[ d(F,v)\right]$. 
The latter holds since instance $I_2$ can be obtained from $I_1$ by $n/4$ consecutive moves of agents actually located in $v$, 
so no move should decrease the expected distance to $v$ (otherwise, an agent in $v$ can lie on its location).

Besides, in instance $I_{2}$, one can place a facility at $v$ which results in the maximum distance of $\sqrt{3}$ to any agent. 
Hence, the approximation ratio is at least:
$$\frac{\E_{F \sim \M(I_{2})} \left[ \max_{p \in I_2} d(F, p) \right]}{OPT(I_{2})}\geq \frac1{\sqrt 3}\cdot\left({\frac{1+2\sqrt{7}}{6} + \sqrt{3}}\right) - O\bigl( \frac 1n \bigr)\approx 1.6054-O(\frac1n).$$

\end{proof}

\subsection{\texorpdfstring{Asymptotically tight lower bound in high-dimensional spaces $\mathbb{R}^{d}$}{Asymptotically tight lower bound in high-dimensional spaces {R}{d}}}

We now focus on high-dimensional spaces. The construction will use regular $d$ simplexes, so we start by definitions and classic properties.

A \emph{regular $d$-simplex} with a side length $R$ is a set
$
S=\operatorname{conv}\{v_0,\dots,v_d\}\subset \mathbb{R}^d
$
such that its $(d+1)$ vertices satisfy
$
\|v_i-v_j\|= R
$
for all $i \ne j$. We also call $R$ the \emph{diameter} of the simplex. 
Equivalently, if the simplex is centered at the origin and has circumradius
$r$, then
$
\sum_{i=0}^d v_i=0,
$
and
$ \|v_i\|=r$ for all $i$. Moreover, 
$ v_i\cdot v_j $
equals $r^{2}$ if $i = j$ and equals $-r^{2}/d$ if $i \ne j$. 
For a simplex $S$, denote $r(S)$ its \emph{radius} and $c(S)$ its \emph{centroid}.
There is a relationship between the radius and the diameter of a regular $d$-simplex: $r(S) = R\sqrt{\frac{d}{2(d+1)}}$.

We now establish a technical result on simplexes, connecting the maximum distance
between any random point $P$ and a simplex vertex to the expected distance from
$P$ to the simplex center. This property will then allow us to get improved
lower bounds on the expected distance between the facility and a simplex vertex
in an iterative construction.

\begin{lemma}	\label{lem:distance-pythagore}
Given a $d$-regular simplex $S$ of radius $r$ with vertices $v_0,\dots,v_d$,
and a random point $P$ of expected distance $R$ to the center of the simplex. 
Then 
\[
\max_{0\le i\le d}\mathbb E\|P-v_i\|
\ge
\sqrt{R^2+r^2} - \frac{4r}{3\sqrt 3\,d}.
\]
In particular, when $d$ is large, this maximum expected distance is asymptotically lower bounded by $\sqrt{R^2+r^2}$.
\end{lemma}
\begin{proof}
W.l.o.g, assume that the simplex is centered at the origin. 
As the simplex is regular, the following properties hold
\[
\|v_i\|=r,
\qquad
\sum_{i=0}^d v_i=0,
\qquad
\frac{1}{d+1}\sum_{i=0}^d v_i v_i^\top
=
\frac{r^2}{d}I_d.
\]

The last equation can be established by the property of simplexes being equiangular tight frames, see for instance~\cite[Definition 1.3 and Example 2 Section 2.2]{sustik2007existence}.

Consider an arbitrary (deterministic) point $p\in \mathbb{R}^d$. We have 
$$
\|p-v_i\|^2 = \|p\|^2+r^2 - 2p\cdot v_i.
$$
Using the algebraic inequality
$
\sqrt{a + b}
\ge
\sqrt a+\frac{b}{2\sqrt a} - \frac{b^2}{2a^{3/2}}
$, which is 
valid for  $a>0$ and $a+b \ge 0$, we get
\[
\|p-v_i\|
\ge
\sqrt{\|p\|^2+r^2 } - \frac{2p\cdot v_i}{2\sqrt{\|p\|^2+r^2 }}
-
\frac{(2p\cdot v_i)^2}{2(\|p\|^2+r^2 )^{3/2}}.
\]
Averaging over $i$, we get
\[
\frac1{d+1}\sum_{i=0}^d \|p-v_i\|
\ge
\sqrt{\|p\|^2+r^2 } 
- 
\frac{1}{(d+1)\sqrt{\|p\|^2+r^2}} \sum_{i=0}^d p\cdot v_i
-
\frac{2}{(d+1)( \|p\|^2+r^2)^{3/2}}\sum_{i=0}^d (p\cdot v_i)^2.
\]
Since
\begin{align*}
\sum_{i=0}^d v_i = 0,
\qquad
\frac1{d+1}\sum_{i=0}^d (p\cdot v_i)^2
=
p^\top
\left(
\frac1{d+1}\sum_{i=0}^d v_i v_i^\top
\right)p
=
\frac{r^2\|p\|^2}{d}
\end{align*}
we deduce
\[
\frac1{d+1}\sum_{i=0}^d \|p-v_i\|
\ge
\sqrt{\|p\|^2+r^2} - \frac{2r^2\|p\|^2}{d(\|p\|^2+r^2)^{3/2}}
\ge
\sqrt{\|p\|^2+r^2}
-
\frac{4r}{3\sqrt 3\,d}.
\]
where the last inequality follows from the algebraic inequality 
$
\frac{2r^2s^2}{(s^2+r^2)^{3/2}}
\le
\frac{4r}{3\sqrt 3}
$.
The latter can be established by setting $y=s/r$ to obtain
$
\frac{2r^2s^2}{(s^2+r^2)^{3/2}}
=
2r\frac{y^2}{(1+y^2)^{3/2}}.
$
Then, note that the function
$
g(y)=\frac{y^2}{(1+y^2)^{3/2}}
$
is maximized at $y=\sqrt 2$, and
$
g(\sqrt 2)=\frac{2}{3\sqrt 3}
$.
Hence
$
\frac{2r^2s^2}{(s^2+r^2)^{3/2}}
\le
2r\cdot \frac{2}{3\sqrt 3}
=
\frac{4r}{3\sqrt 3}.
$

Therefore, for every fixed $p$,
\[
\frac1{d+1}\sum_{i=0}^d \|p-v_i\|
\ge
\sqrt{\|p\|^2+r^2}
-
\frac{4r}{3\sqrt 3\,d}.
\]

Now let $P$ be a random point. Applying the previous pointwise
inequality to $P$ and taking expectations gives
\begin{align*}
\frac1{d+1}\sum_{i=0}^d \mathbb E\|P-v_i\|
\ge
\mathbb E\sqrt{\|P\|^2+r^2} - \frac{4r}{3\sqrt 3\,d} 
&\ge 
\sqrt{(\mathbb E\|P\|)^2+r^2} - \frac{4r}{3\sqrt 3\,d} \\
&= \sqrt{R^2+r^2} - \frac{4r}{3\sqrt 3\,d}
\end{align*}
where the last inequality follows Jensen's inequality applied to the convex function 
$s\mapsto \sqrt{s^2+r^2}$ on $[0,\infty)$. 
As the maximum is at least the average, the lemma follows.
\end{proof}

We are now ready to prove Theorem~\ref{thm:lower-bound-of-2}. The construction
used is similar to the previous one on the plane, except that each step creates
a new orthogonal $d$-simplex and more steps are considered. For large $d$, the
ratio between the simplex edge length and its radius is then maximized, which
ultimately leads to better guarantees. Together with the pairwise orthogonality
between simplexes, this allows to lower the radius of the final hypersphere.
Increasing the number of rounds allows to increase the final expected distance
between an agent and the facility, matching asymptotically the hypersphere
radius. 

\begin{theorem}	\label{thm:lower-bound-of-2}
For sufficiently large parameters $d, n$, 
any strategy-proof mechanism in $\mathbb{R}^{d\log d}$ with $n$ agents 
has a competitive ratio at least 
$$
2 - O\biggl( \frac{1}{d} + \frac{1}{n^{1/d\log d}} \biggr).
$$
In other words, for $\varepsilon > 0$ arbitrarily small, by choosing $d = \Omega(1/\varepsilon)$ and
$n = \Omega\bigl( (1/\varepsilon)^{d \log d} \bigr)$, the approximation ratio is at least 
$2 - \varepsilon$.
\end{theorem}
\begin{proof}
Let $k, d, n$ be large parameters to be defined later. 
Consider a strategy-proof mechanism $\M$ in $\mathbb R^{kd}$.
Let $(e_{1},\dots,e_{kd})$ be an orthonormal basis of $\mathbb R^{kd}$. 
Recall that a regular $d$-simplex $S$ of side length $R$ has radius $r(S)= R\sqrt{\frac{d}{2(d+1)}} =  R/\sqrt{2} - O(R/d)$
and $c(S)$ is the centroid of $S$.

Let  $S_1$ be a regular $d$-simplex of side length $2^{1/2}$ centered at the origin, named $A_0$,
and $S_{1}$ lies in the space spanned by vectors $(e_{1},\dots,e_{d})$. 
So $r(S_{1}) = 1 - O(2^{1/2}/d)$.
Consider an instance $I_{1}$ in which there are $n$ agents on each vertex of simplex $S_{1}$.  
Then, by \Cref{lem:distance-pythagore}, there exists a vertex in simplex $S_{1}$, say $A_{1}$, 
of expected distance to the facility lower-bounded by the following.
\begin{align*}
\E_{F \sim \M(I_{1})} \left[ d(F, A_{1}) \right] 
&= \max_{v_{i} \in S_{1}} \E_{F \sim \M(I_{1})} \left[ d(F, v_{i}) \right] \\
&\geq \sqrt{\E\left[d(F, c(S_{1})) \right]^2+r(S_{1})^2} - \frac{4r(S_{1})}{3\sqrt 3\,d} \\
&\geq \sqrt{r(S_{1})^2} - \frac{4r(S_{1})}{3\sqrt 3\,d}
\geq 1 - O\biggl(\frac{r(S_{1})}{d}\biggr)
\end{align*}  

See Figure~\ref{fig:LBd} for an illustration.

\begin{figure}[htbp]
    \centering
    \begin{subfigure}[b]{0.24\textwidth}
        \centering
        \includegraphics[width=\textwidth, page=1]{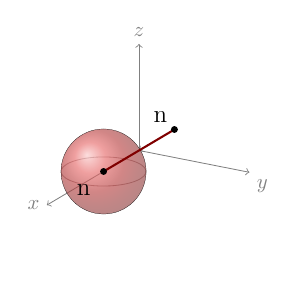}
        \caption{Instance $I_1$}
    \end{subfigure}
    \begin{subfigure}[b]{0.24\textwidth}
        \centering
        \includegraphics[width=\textwidth, page=2]{fig/plot_LB.pdf}
        \caption{Instance $I_2$}
    \end{subfigure}
    \begin{subfigure}[b]{0.24\textwidth}
        \centering
        \includegraphics[width=\textwidth, page=3]{fig/plot_LB.pdf}
        \caption{Instance $I_3$}
    \end{subfigure}
    \begin{subfigure}[b]{0.24\textwidth}
        \centering
        \includegraphics[width=\textwidth, page=4]{fig/plot_LB.pdf}
        \caption{Final instance $I$}
    \end{subfigure}
    \caption{Illustration of the lower bound construction with $d=1$ and $k=3$. The segments represent the simplexes $S_i$, the first three spheres centered on $A_1$, $A_2$, $A_3$  represent the minimum average distance between the facility location and the sphere center in the corresponding instance. The actual construction requires $d>1$ and 4 dimensions to have spheres of increasing radius. The last sphere represent the hypersphere $\mathbb S$. The numbers next to points and $\mathbb S$ represents the number of agents at this location. For large $d$ and $k$, the sphere centered on $A_k$ and $\mathbb S$ are indistinguishable.}
    \label{fig:LBd}
\end{figure}

Let $S_{2}$ be a $d$-regular simplex of side length $2^{2/2}$ centered at $A_1$ 
and  lying in the space spanned by vectors $(e_{d+1},\dots,e_{2d})$. Note that $r(S_{2}) = 2^{1/2} - O(2^{2/2}/d)$.
Intuitively, simplex $S_{2}$ is orthogonal to simplex $S_{1}$. 
Consider the instance $I_{2}$ in which there are $n$ agents on all vertices of $S_{1}$ except $A_{1}$ 
and $n/(d+1)$ agents on each vertex of $S_{2}$.
Then there exists a vertex $S_{2}$, say $A_{2}$, having the distance to the facility 
at least the average distance from vertices in $S_{2}$ to the facility location $\M(I_{2})$. 
Specifically, 
\begin{align*} 
\E_{F \sim \M(I_{2})} &\left[ d(F, A_{2}) \right] 
= \max_{v_{i} \in S_{2}} \E_{F \sim \M(I_{2})} \left[ d(F, v_{i}) \right]  \\
&\geq \sqrt{r(S_2)^2+ \left( \E_{F \sim \M(I_{2})} \left[ d(F, c(S_{2}))\right] \right)^{2}}  - \frac{4r(S_{2})}{3\sqrt 3\,d}  \\
&\geq \sqrt{r(S_2)^2+ \left( \E_{F \sim \M(I_{1})} \left[ d(F,A_{1})\right] \right)^{2}} - \frac{4r(S_{2})}{3\sqrt 3\,d}\\
&\geq \sqrt{2 - O(2^{2/2}/d)+1 - O(r(S_{1})/d)}  - O(r(S_{2})/d) \\
&\geq \sqrt{3} - O\bigl( (r(S_{1}) + r(S_{2}))/d \bigr).
\end{align*}
where
\begin{itemize}
	\item the first inequality is by \Cref{lem:distance-pythagore};
	\item the second inequality holds by the strategy-proofness of mechanism $\M$ and note that the center $c(S_{2})$ is $A_{1}$. 
		Specifically, it must hold that  $ \E_{F \sim \M(I_{2})} \left[ d(F,A_{1})\right] \geq  \E_{F \sim \M(I_{1})} \left[ d(F,A_{1})\right] $ 
		since instance $I_2$ can be obtained from $I_1$ by $n$ consecutive moves of agents actually located in $A_1$, 
		so no move should decrease the expected distance to $A_1$ (otherwise, an agent in $A_{1}$ can lie on its location).
	\item the last inequality follows $\E_{F \sim \M(I_{1})} \left[ d(F,A_{1})\right] \geq 1 - O(r(S_{1})/d)$.
\end{itemize}

We repeat this process $k$ steps and prove the following by induction. 
At step $\ell \geq 1$, let $A_{\ell}$ be a vertex in the $d$-regular simplex $S_{\ell}$ 
that have the maximum expected distance to facility $F \sim \M(I_{\ell})$. 
Assume that 
$\E_{F \sim \M(I_{\ell})} \left[ d(F, A_{\ell}) \right] \geq \sqrt{2^{\ell} - 1} - O\bigl( (r(S_{1}) + r(S_{2}) + \ldots + r(S_{\ell}))/d \bigr)$.

Now let $S_{\ell+1}$ be a $d$-regular simplex of side length $2^{(\ell+1)/2}$ centered at $A_{\ell}$ 
and $S_{\ell+1}$ lies in the space spanned by vectors $(e_{\ell d+1},\dots,e_{(\ell+1) d})$. 
Intuitively, simplex $S_{\ell+1}$ is orthogonal to other simplexes $S_{\ell}, \ldots, S_{1}$. 
Moreover, $r(S_{\ell + 1}) = 2^{\ell/2} - O\bigl( 2^{(\ell+1)/2}/d\bigr)$, i.e., the distance from center $A_{\ell}$ to any vertex of $S_{\ell+1}$ is at most $2^{\ell/2}$.
Consider the instance $I_{\ell+1}$ in which there are the same sets of agents 
as in instance $I_{\ell}$ except that
there is no agent on vertex $A_{\ell}$ and 
there are additionally $n/(d+1)^{\ell}$ agents on each vertex of $S_{\ell+1}$. 
We will prove that there exists a vertex 
$A_{\ell+1}$ in simplex $S_{\ell+1}$ such that 
\begin{align}	\label{ineq:lb-2-recurrent}
\E_{F \sim \M(I_{\ell+1})} \left[ d(F, A_{\ell+1}) \right] \geq \sqrt{2^{\ell + 1} - 1} - O\bigl( (r(S_{1}) + \ldots + r(S_{\ell}) + r(S_{\ell+1}))/d \bigr).
\end{align}

Specifically, let $A_{\ell+1}$ be a vertex in simplex $S_{\ell+1}$ with the maximum expected distance to the facility 
$F \sim \M(I_{\ell+1})$. We have:
\begin{align*} 
\E_{F \sim \M(I_{\ell + 1})} &\left[ d(F, A_{\ell+1}) \right] 
= \max_{v_{i} \in S_{\ell + 1}} \E_{F \sim \M(I_{\ell + 1})} \left[ d(F, v_{i}) \right] \\
&\geq \sqrt{r(S_{\ell + 1})^2 + \left( \E_{F \sim \M(I_{\ell+1})} \left[ d(F,c(S_{\ell + 1}))\right] \right)^{2}} - \frac{4r(S_{\ell + 1})}{3\sqrt 3\,d} \\ 
&\geq \sqrt{r(S_{\ell + 1})^2 + \left( \E_{F \sim \M(I_{\ell})} \left[ d(F,A_{\ell})\right] \right)^{2}} - \frac{4r(S_{\ell + 1})}{3\sqrt 3\,d} \\
&\geq \sqrt{2^{\ell} - O\bigl( 2^{(\ell+1)/2}/d\bigr) + 2^{\ell} - 1 - O\bigl(2^{\ell/2} (r(S_{1}) + \ldots + r(S_{\ell}))/d \bigr)} - O(r(S_{\ell+1})/d) \\
&\geq \sqrt{2^{\ell+1} - 1} - O\bigl( (r(S_{1}) + \ldots + r(S_{\ell}) + r(S_{\ell+1}))/d \bigr),
\end{align*}

where
\begin{itemize}
	\item the first inequality is due to \Cref{lem:distance-pythagore};
	\item the second inequality holds by the strategy-proofness of mechanism $\M$ and note that the center $c(S_{\ell+1})$ is $A_{\ell}$. 
		Specifically, it must hold that  $\E_{F \sim \M(I_{\ell + 1})} \left[ d(F,A_{\ell})\right] \geq  \E_{F \sim \M(I_{\ell})} \left[ d(F,A_{\ell})\right] $ 
		since instance $I_{\ell+1}$ can be obtained from $I_\ell$ by $n/(d+1)^{\ell-1}$ consecutive moves of agents actually located in $A_\ell$, 
		so no move should decrease the expected distance to $A_\ell$ (otherwise, an agent in $A_{\ell}$ can lie on its location).
	\item the last inequality follows from \Cref{ineq:lb-2-recurrent}.
\end{itemize}

We stop the process after $k$ steps. 
Now consider a hypersphere $\mathbb{S}$ centered at $A_{k}$ of radius $2^{k/2}$ --- side length (or diameter) of $S_{k}$. 

We claim that all vertices of simplexes $\cup_{\ell = 1}^{k} S_{\ell}$ belong to $\mathbb{S}$. 
Recall that, by our construction, simplex $S_{\ell}$ is orthogonal to all other simplexes. 
By Pythagore theorem, for any simplex $S_{\ell}$ for $1 \leq \ell \leq k$ and for any vertex $p \in S_{\ell}$, we have 
\begin{align}	\label{ineq:lb-2-opt}
d(p,A_{k})^{2} 
&= d(p,A_{\ell})^2 + d(A_{\ell},A_{\ell+1})^2 + \ldots + d(A_{k-1},A_{k})^2	\notag \\
&= \text{diam}(S_\ell)^2 + r(S_{\ell+1})^2 + \ldots + r(S_{k})^2 \notag \\
&\leq 2^{\ell +1}  + 2^{\ell+1} + \dots +2^{k-1} 
\leq 2^{k}.
\end{align}

Consider the final instance $I$ obtained from $I_{k}$ by spreading agents on $A_{k}$ 
uniformly on the surface of the hypersphere $\mathbb{S}$ 
(all agents in $I_{k}$ located in positions other than $A_{k}$ remain at their place). 

By the observation from \Cref{ineq:lb-2-opt}, 
an optimal solution of instance $I$ can place the facility at location $A_{k}$ which results in the 
objective (maximum distance to any agent) of $2^{k/2}$. 

We are now bounding the expected objective of mechanism $\M$ on instance $I$.
By the strategy-proofness of $\M$, we have 
\begin{align}	\label{ineq:facility-center}
\E_{F \sim \M(I)} \left[ d(F, A_{k}) \right] 
&\geq \E_{F \sim \M(I_{k})} \left[ d(F, A_{k}) \right] 	\notag \\
&\geq \sqrt{2^{k} - 1} - O\bigl( (r(S_{1}) + \ldots + r(S_{k}))/d \bigr).
\end{align}
Observe that, in instance $I$, $n/(d+1)^{k}$ agents are spread out uniformly on the sphere
$\mathbb{S}$ so any point in the hypersphere $\mathbb{S}$ has a distance at most of $O(2^{k/2} \cdot (n/d^{k})^{-1/kd})$ to 
a location of an agent (where recall that $2^{k/2}$ is the radius of the sphere).    
Hence, for any sample of location of facility $F$, there exists an agent 
located at location $p_{F} \in \mathbb{S}$ in the opposite side of $F$ through the center $A_{k}$ of the sphere such that the distance 
$d(F, p_{F}) = d(F, A_{k}) + d(A_{k}, p_{F}) - O(2^{k/2} \cdot (n/d^{k})^{-1/kd})$. 
Therefore, the expected objective of the mechanism $\M$ is 
\begin{align*}
&\E_{F \sim \M(I)} \left[ \max_{p \in I} d(F, p) \right]
\geq \E_{F \sim \M(I)} \left[ d(F, p_{F}) \right] \\
&\geq \E_{F \sim \M(I)} \left[ d(F, A_{k}) + d(A_{k}, p_{F}) - O\bigl( 2^{k/2} \cdot (n/d^{k})^{-1/kd} \bigr) \right] \\
&= \E_{F \sim \M(I)} \left[ d(F, A_{k}) \right] + 2^{k/2}  - O\bigl( 2^{k/2} \cdot (n/d^{k})^{-1/kd} \bigr) \\
&\geq  \sqrt{2^{k}-1} - O\bigl( (r(S_{1}) + \ldots + r(S_{k}))/d \bigr) + 2^{k/2}  - O\bigl( 2^{k/2} \cdot (n/d^{k})^{-1/kd} \bigr) \\
&\geq 2\cdot 2^{k/2} - O\bigl(2^{-k/2}+2^{(k+1)/2}/d +  2^{k/2} \cdot (n/d^{k})^{-1/kd} \bigr) 
\end{align*}
where the third inequality is due to \Cref{ineq:facility-center} and in the last inequality, we use $r(S_{\ell}) \leq 2^{\ell/2}$.

Hence, the approximation ratio of mechanism $\M$ is at least: 
\begin{align*}
\frac{2\cdot 2^{k/2} - O\bigl(2^{-k/2}+2^{(k+1)/2}/d +  2^{k/2} \cdot (n/d^{k})^{-1/kd} \bigr) }{2^{k/2}}
&= 2 - O\biggl(2^{-k} + \frac{1}{d} + \frac{n^{-1/kd}}{d^{-1/d}} \biggr) \\
&= 2 - O\biggl(2^{-k} + \frac{1}{d} + \frac{1}{n^{1/kd}} \biggr).
\end{align*}

For $\varepsilon > 0$ arbitrarily small, by choosing $d = \Omega(1/\varepsilon)$, $k = \log d$ and 
$n = \Omega\bigl( (1/\varepsilon)^{d \log d} \bigr)$, the approximation ratio is at least 
$2 - \varepsilon$.

\end{proof}

\section{\texorpdfstring{Better algorithms for $n$ agents in Euclidean spaces}{Better algorithms for n agents in Euclidean spaces}}

The best known mechanism for a large number of agents is the \emph{Centroid} mechanism: with a probability $1/2$, the mechanism outputs the centroid of the reported locations (arithmetic average), and with probability $1/2$, it outputs one agent location drawn uniformly at random (so a probability $1/2n$ per agent)~\cite{tang2020characterization}. It achieves an approximation factor of $2-1/n$ for $n$ agents.

In this section, we introduce novel mechanisms that improve over this bound, focusing first on the Euclidean plane, then allowing an additional dimension for the facility while the agents stay on the plane, and finally allowing this additional dimension when the agents now live in any $d$-dimensional Euclidean space.

\subsection{The \emph{Box-center} mechanism on the Euclidean plane}
\label{sec:R2R2UB}

We consider the following \emph{Box-center} mechanism: with a probability $2/(n+2)$, the mechanism outputs the center of the axis-aligned bounding box of the reported locations, and outputs each agent with probability $1/(n+2)$.

Using the bounding box to design a mechanism with the egalitarian objective has not been used successfully up to our knowledge until~\cite{balkanski2024randomized} considered a learning-augmented variant of the problem, where the hint on the extreme points allow to mitigate the impact of interior agent lies. We show here that outputing a \emph{center} with a probability much lower than the value of $1/2$ used in the centroid mechanism actually leads to an improved performance for $n>3$. The rationale is that the box center is generally much closer to the optimal facility location, but is more impacted by agent lies.

\begin{theorem}
\label{th:boxcenterR2}
The \emph{Box-center} mechanism on $\mathbb R^2$ is strategyproof and has an approximation factor at most $2-\frac{3-\sqrt 2}{n+2}$.
\end{theorem}

\begin{proof}

We first show that the \emph{Box-center} mechanism is strategyproof. Denote the agent true locations at $\vect{x} = (x_{1}, \ldots, x_{n})$. If agent $i$ lies to $x'_i$, the box-center position (chosen by the mechanism with probability $2/(n+2)$) is modified by a distance at most $\frac 12||x_i-x'_i||$. The other agent locations are unchanged, and the agent $i$ reported location (chosen by the mechanism with probability $1/(n+2)$) is then by definition at a distance exactly $||x_i-x'_i||$ away from $x_i$. Therefore, the expected distance of the reported facility to $x_i$ did not decrease as $\frac1{n+2}||x_i-x'_i||\geq \frac2{n+2}\cdot \frac{1}{2}||x_i-x'_i||$,  establishing strategyproofness.

If $OPT=0$, all reports coincide and the mechanism is optimal. Otherwise, we now show an upper bound on the approximation factor. We assume without loss of generality that the smallest enclosing disk of the agent locations is the unit disk. We upper bound the distance between the axis-aligned bounding box center $C=(c_1,c_2)$ and the furthest agent $x_i$ at coordinates $(a,b)$, see Figure~\ref{fig:boxcenterR2}. Note that $a^2+b^2\leq1$ by the enclosing disk assumption. Assume as well, by flipping the axis directions, that $a$ and $b$ are nonnegative. All reported coordinates belong to $[-1,1]$, while $(a,b)$
is itself a reported location. Since $a,b\geq 0$, it follows that
$|a-c_1|\leq (1+a)/2$ and $|b-c_2|\leq (1+b)/2$. Therefore,
\begin{align*}
d^*(C)
&\leq \sqrt{\frac{(1+a)^2}{4}+\frac{(1+b)^2}{4}}.
\end{align*}

Under the constraint $a^2+b^2\leq 1$, we have
\[(1+a)^2+(1+b)^2 =2+(a^2+b^2)+2(a+b) \leq 3+2\sqrt{2(a^2+b^2)} \leq 3+2\sqrt{2} .\]
Equality is attained for $a=b=1/\sqrt{2}$. 
Therefore,
\[
d^*(C)
\leq \frac{1}{2}\sqrt{3+2\sqrt{2}}
= \frac{1}{2}\sqrt{(1+\sqrt{2})^2}
= \frac{1+\sqrt{2}}{2}.
\]

As the maximum distance between any two agents is at most $2$, we get that the expected distance between the output facility location and the furthest agent is at most:
\[
d^* \leq \frac{2}{n+2}\cdot d^*(C)+\frac{n}{n+2}\cdot 2  \leq \frac{1+\sqrt2+2n}{n+2}  = 2-\frac{3-\sqrt 2}{n+2}.
\]
\begin{figure}
\centering
\includegraphics[width=.4\linewidth]{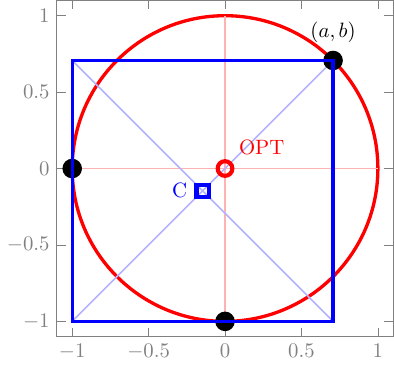}
\caption{Illustration of the \emph{Box-center} mechanism: the agents in dots, the enclosing circle in red and the axis-aligned bounding box in blue, in the case where the bounding box center $C$ is the furthest from the optimal facility location.}
\label{fig:boxcenterR2}
\end{figure}
\end{proof}

\begin{corollary}
The \emph{Box-center} mechanism on $\mathbb R^3$ is strategyproof and has an approximation factor at most $2-\frac{3-\sqrt 3}{n+2}$.
\end{corollary}

\begin{proof}
The  proof of Theorem~\ref{th:boxcenterR2} can be directly adapted, replacing the coordinates of the furthest agent $x_i$ by $(a,b,c)$ with $||x_i||^2\leq1$. Similarly to the case $d=2$, the maximum is obtained with $a=b=c=\frac{1}{\sqrt{3}}$, leading to: $$d^*(C)\leq\sqrt{\frac{(a+1)^2}{4}+\frac{(b+1)^2}{4}+\frac{(c+1)^2}{4}}\leq\frac12 \cdot\sqrt{3}\cdot \left(\frac1{\sqrt 3}+1\right)=\frac{1+\sqrt3}2. $$
\end{proof} 

\subsection{\texorpdfstring{The \emph{Lifted Box-center} mechanism for planar agents and $\mathbb R^3$ facility}{The {Lifted Box-center} mechanism for planar agents and  R3 facility}}
\label{sec:r2-r3}

We focus in this section on the case $\II=\mathbb R^2$ and $\OO=\mathbb R^3$.
Reports remain in the horizontal plane: a location $x_i \in\mathbb R^2$ is
identified with $(x_i,0)\in\mathbb R^3$, while the mechanism may place the
facility above that plane.

We reuse the bounding box at the basis of the previous algorithm, but exploit the additional dimension to penalize the spread of the agents, instead of drawing agents uniformly at random as previously. Specifically, the facility is placed at the bounding box center, lifted by a distance equal to half the bounding box diagonal. We show that such a mechanism is truthful and achieves an approximation factor of $\sqrt{3}$. Unfortunately, this idea generalized to higher dimensions leads to higher approximation ratios, worse than the classic $2-1/n$ upper bound.

We formally define the \emph{Lifted Box-center} mechanism as follows, see \Cref{fig:liftedbox} for an illustration.
    
For $n\ge2$ reports $x_i=(x_{i,1},x_{i,2})$, define, for
$k\in\{1,2\}$,
\[
  L_k=\min_i x_{i,k},
  \qquad
  U_k=\max_i x_{i,k}.
\]
Let
\[
  c_k=\frac{L_k+U_k}{2},
  \qquad
  w_k=\frac{U_k-L_k}{2},
  \qquad
  c=(c_1,c_2),
  \qquad
  h_{\rm box}=\sqrt{w_1^2+w_2^2}.
\]
The Lifted box mechanism outputs:
\[
  F=(c,h_{\rm box})\in\mathbb R^3.
\]

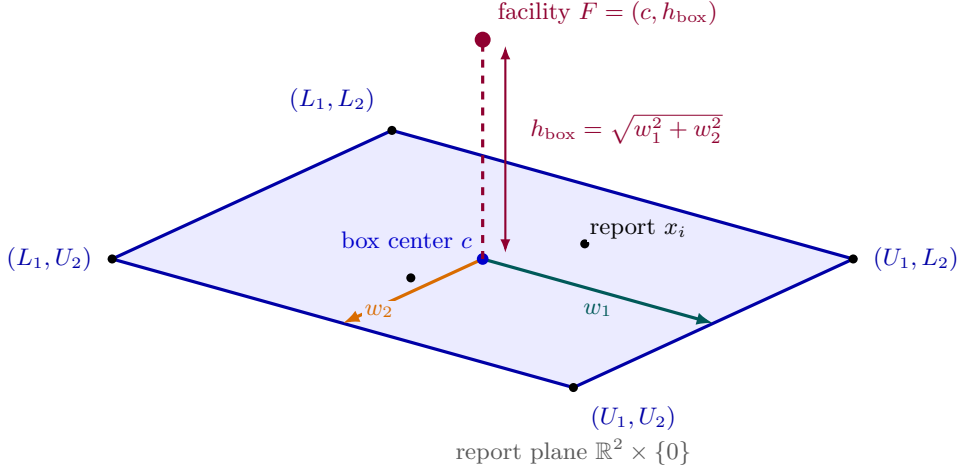
\begin{figure}
\begin{tikzpicture}[>=Latex,every node/.style={font=\small}]
  \coordinate (A) at (3.00,4.05);
  \coordinate (B) at (9.10,2.35);
  \coordinate (D) at (-.70,2.35);
  \coordinate (E) at (5.40,.65);
  \coordinate (C) at (4.20,2.35);
  \coordinate (Mone) at (7.25,1.50);
  \coordinate (Mtwo) at (2.35,1.50);
  \coordinate (F) at (4.20,5.25);

  \fill[blue!8] (A)--(B)--(E)--(D)--cycle;
  \draw[blue!65!black,very thick] (A)--(B)--(E)--(D)--cycle;

  \foreach \P in {
    (3.00,4.05),
    (9.10,2.35),
    (-.70,2.35),
    (5.40,.65),
    (3.25,2.10)
  }
    \fill[black] \P circle (1.7pt);

  \coordinate (Ri) at (5.55,2.55);
  \fill[black] (Ri) circle (1.7pt)
    node[above right=1pt,fill=blue!8,inner sep=1pt]
    {report $x_i$};

  \node[blue!65!black,anchor=south east]
    at ($(A)+(-.10,.12)$) {$(L_1,L_2)$};

  \node[blue!65!black,anchor=west]
    at ($(B)+(.14,0)$) {$(U_1,L_2)$};

  \node[blue!65!black,anchor=east]
    at ($(D)+(-.14,0)$) {$(L_1,U_2)$};

  \node[blue!65!black,anchor=north west]
    at ($(E)+(.10,-.12)$) {$(U_1,U_2)$};

  \draw[teal!70!black,very thick,-{Latex[length=2.2mm]}]
    (C)--(Mone)
    node[midway,below=3pt,fill=blue!8,inner sep=1pt] {$w_1$};

  \draw[orange!85!black,very thick,-{Latex[length=2.2mm]}]
    (C)--(Mtwo)
    node[pos=.58,below left=2pt,fill=blue!8,inner sep=1pt] {$w_2$};

  \fill[blue!80!black] (C) circle (2.2pt)
    node[above left=4pt,fill=blue!8,inner sep=1pt]
    {box center $c$};

  \draw[purple!75!black,very thick,dashed] (C)--(F);

  \draw[
    purple!75!black,
    thick,
    {Latex[length=1.8mm]}-{Latex[length=1.8mm]}
  ]
    ($(C)+(.30,.08)$)--($(F)+(.30,-.08)$)
    node[pos= 0.6,right=7pt,fill=white,inner sep=2pt]
    {$h_{\rm box}=\sqrt{w_1^2+w_2^2}$};

  \fill[purple!75!black] (F) circle (3pt)
    node[above right=3pt]
    {facility $F=(c,h_{\rm box})$};

  \node[gray!70!black,anchor=north]
    at ($(E)+(0,-.55)$)
    {report plane $\mathbb{R}^2\times\{0\}$};
\end{tikzpicture}
\caption{Illustration of the \emph{Lifted Box-center} mechanism.}
\label{fig:liftedbox}

\end{figure}

\begin{theorem}
\label{thm:lifted-box}
The \emph{Lifted Box-center} mechanism is deterministic truthful and has
approximation ratio $\sqrt3$ .
\end{theorem}

\begin{proof}
We first prove truthfulness. Fix an agent $i$, with true location
$x_i=(x_{i,1},x_{i,2})$, and fix the reports of all other agents. Let
$(c_1,c_2)$ be the center of the bounding box and let $w_1,w_2$ be its
half-widths. Since the facility is lifted by a height
$\sqrt{w_1^2+w_2^2}$, the squared distance from the agent to the facility is
\[
    \|(x_i,0)-\M(\vect{x})\|^2
    =
    (x_{i,1}-c_1)^2+w_1^2
    +(x_{i,2}-c_2)^2+w_2^2.
\]
The two coordinates therefore contribute independently to the squared
distance. We can thus consider a deviation in one coordinate at a time, starting with the first. Recall $
    c_1=\frac{L_1+U_1}{2},
    \quad
    w_1=\frac{U_1-L_1}{2}.
   $
   
The contribution of this coordinate to the squared distance is
\begin{equation}
\label{eq:box-coordinate-energy}
    (x_{i,1}-c_1)^2+w_1^2
    =
\frac{(x_{i,1}-c_1+w_1)^2+(x_{i,1}-c_1-w_1)^2}{2}
=
    \frac{(x_{i,1}-L_1)^2+(x_{i,1}-U_1)^2}{2}.
\end{equation}

This identity shows that truthful reporting minimizes the agent's cost,
since it minimizes both $|x_{i,1}-L_1|$ and $|x_{i,1}-U_1|$.
Indeed, if $L_1<x_{i,1}<U_1$, any lie that changes the box can only
decrease $L_1$ or increase $U_1$, thereby increasing the contribution of
this coordinate. Similarly, if the agent is on the boundary, say
$x_{i,1}=L_1$, then $|x_{i,1}-L_1|=0$ cannot be decreased, while moving
the report away from $x_{i,1}$ can only increase one of the two terms.

The same argument applies to a deviation in the second coordinate. Finally,
any deviation in the plane can be decomposed into two deviations,
first changing one coordinate and then the other. Neither step can decrease
the agent's cost, so no deviation can be beneficial. Therefore the
Lifted Box-center mechanism is truthful.\\

\begin{figure}[t]
\centering
\resizebox{\linewidth}{!}{%
\begin{tikzpicture}[>=Latex,every node/.style={font=\small}]
  \coordinate (O) at (0,0);
  \coordinate (c) at (.33,.33);
  \coordinate (Xi) at (2.25,0);
  \coordinate (LBox) at (-1.59,-1.59);
  \coordinate (UBox) at (2.25,2.25);

  \fill[gray!5] (O) circle (2.25);
  \draw[gray!35,thin] (-2.38,0)--(2.38,0);
  \draw[gray!35,thin] (0,-2.38)--(0,2.38);

  \fill[blue!20,fill opacity=.12] (LBox) rectangle (UBox);

  \draw[gray!70,very thick] (O) circle (2.25);
  \draw[blue!65!black,very thick] (LBox) rectangle (UBox);

  \foreach \Q in {
    (-1.59,-1.59),
    (1.59,1.59),
    (2.25,0),
    (0,2.25)
  }
    \fill[black] \Q circle (1.7pt);

  \draw[gray!55,densely dashed]
    (-1.59,-1.59)--(1.59,1.59);

  \node[gray!70!black,fill=white,inner sep=2pt]
    at (-1,2.55)
    {normalized smallest enclosing disk:
     $\operatorname{OPT}=R=1$};

  \node[blue!65!black,fill=white,inner sep=1pt]
    at (-1.38,-1.82)
    {bounding box};

  \fill[black] (O) circle (1.6pt)
    node[below left=2pt] {$0$};

  \draw[red!70!black,very thick,-{Latex[length=2mm]}]
    (O)--(c);

  \fill[blue!75!black] (c) circle (2.1pt)
    node[above left=4pt,fill=white,inner sep=1pt] {$c$};

  \fill[black] (Xi) circle (1.9pt)
    node[right=5pt,fill=white,inner sep=1pt] {$x_i$};

  \draw[purple!75!black,very thick,-{Latex[length=2mm]}]
    (c)--(Xi)
    node[midway,above=5pt,fill=white,inner sep=1pt]
    {$\delta_i=x_i-c$};

  \draw[
    teal!70!black,
    very thick,
    {Latex[length=1.7mm]}-{Latex[length=1.7mm]}
  ]
    (.33,.82)--(2.25,.82)
    node[midway,above=2pt,fill=white,inner sep=1pt]
    {$w_1$};

  \draw[densely dashed,teal!55] (.33,.33)--(.33,.82);
  \draw[densely dashed,teal!55] (2.25,0)--(2.25,.82);

  \draw[
    orange!85!black,
    very thick,
    {Latex[length=1.7mm]}-{Latex[length=1.7mm]}
  ]
    (-.18,.33)--(-.18,2.25)
    node[midway,left=3pt,fill=white,inner sep=1pt]
    {$w_2$};

  \draw[densely dashed,orange!60] (-.18,.33)--(.33,.33);
  \draw[densely dashed,orange!60] (-.18,2.25)--(0,2.25);

  \node[
    anchor=west,
    draw=black!30,
    rounded corners=2pt,
    fill=gray!3,
    inner sep=7pt,
    align=left
  ] (ineq) at (3.05,.90) {$\displaystyle
    \begin{aligned}
      w_1&\le1-c_1,\\
      w_2&\le1-c_2,\\[2pt]
      h_{\rm box}^2
        &=w_1^2+w_2^2\\
        &\le2-2(c_1+c_2)+\lVert c\rVert^2,\\[2pt]
      \lVert\delta_i\rVert^2
        &\le1+\lVert c\rVert^2-2x_i\mathbin{\cdot}c,\\
      x_i\mathbin{\cdot}c
        &\ge2\lVert c\rVert^2-(c_1+c_2).
    \end{aligned}$};

  \node[
    anchor=north west,
    draw=purple!45!black,
    very thick,
    rounded corners=2pt,
    fill=purple!4,
    inner sep=7pt,
    align=center
  ] at ($(ineq.south west)+(0,-.30)$) {$\displaystyle
    \begin{aligned}
      \lVert(x_i,0)-M_{\rm box}(x)\rVert^2
        &=\lVert\delta_i\rVert^2+h_{\rm box}^2\\
        &\le3-2\lVert c\rVert^2\\
        &\le3.
    \end{aligned}$};

\end{tikzpicture}%
}
\caption{The ratio argument for the Lifted Box-center mechanism. The smallest
enclosing disk has been normalized to radius one.}
\label{fig:lifted-box-ratio}
\end{figure}

Now focusing on the approximation ratio, let $R$ be the radius of a smallest enclosing disk of the
reports.  If $R=0$, all reports coincide and the mechanism returns their
common location at height zero.  Assume $R>0$, and translate and scale the
smallest enclosing disk to the unit disk centered at the origin.  By flipping
axis directions, we can assume $c_1,c_2\ge0$.  Since the bounding box lies in
$[-1,1]^2$, we get
\[
  w_1\le1-c_1,
  \qquad
  w_2\le1-c_2.
\]
Figure~\ref{fig:lifted-box-ratio} summarizes the whole argument: fix a reported point $x_i$, and write $\delta_i=x_i-c$.  First,
\begin{equation}
\label{eq:box-height-bound}
  w_1^2+w_2^2
  \le 2-2(c_1+c_2)+\|c\|^2.
\end{equation}
Second, $\|x_i\|\le1$ gives
\begin{equation}
\label{eq:box-planar-bound}
  \|\delta_i\|^2=\|x_i-c\|^2
  \le1+\|c\|^2-2x_i\mathbin{\cdot}c.
\end{equation}
Every point of the box satisfies
\[
  x_{i,1}\ge c_1-w_1\ge2c_1-1,
  \qquad
  x_{i,2}\ge c_2-w_2\ge2c_2-1.
\]
Because $c_1,c_2\ge0$, it follows that
\[
  x_i\mathbin{\cdot}c = x_{i,1} c_1 +x_{i,2} c_2 \geq 2c_1^2 - c_1 + 2 c_2^2 -c_2
  = 2\|c\|^2-(c_1+c_2).
\]
Combining this inequality with
\Cref{eq:box-height-bound,eq:box-planar-bound} yields
\begin{align*}
  \|(x_i,0)-\M(\vect{x})\|^2
  &=\|\delta_i\|^2+w_1^2+w_2^2 \\
  &\leq 1+\|c\|^2-2x_i\mathbin{\cdot}c  + 2-2(c_1+c_2)+\|c\|^2 \\
  &\le3+2\|c\|^2-4\|c\|^2
  \le3.
\end{align*}
Since the smallest enclosing disk has been normalized to radius one, the ratio is $\sqrt3$. 

To see that this approximation factor is tight, consider the four-point cross
$\{
  (1,0),\allowbreak (-1,0),\allowbreak(0,1),\allowbreak(0,-1)
\}$.
Its optimum radius is $1$, while $c=0$, $w_1=w_2=1$, and the lifted
height is $h_{\rm box}=\sqrt2$.  Every boundary agent has cost $\sqrt3$. 
\end{proof} 

\subsection{The \emph{SD-lift} mechanism exploiting an additional dimension}

\label{sec:RdRdp}

Assume the agents live in a euclidean space of $d$ dimensions $\mathbb R^d$ for any $d$, and the mechanism has access to an additional orthogonal dimension. Consider the deterministic \emph{SD-lift} mechanism, in which the mechanism places the facility at the centroid of the reported locations, lifted by a distance equal to the standard deviation of the reported location in the additional dimension. Specifically, if the agents reported location are at $\vect x = (x_1,\dots, x_n)$, and the centroid is $\bar x=\frac1n\sum_{i=1}^nx_i$, then the lift distance is equal to $\sigma = \sqrt{\frac1n\sum_{i=1}^n||x_i-\bar x||^2}$.

\begin{theorem}
The \emph{SD-lift} mechanism is truthful and has an approximation factor of $2\sqrt{\frac{n-1}{n}}$.
\end{theorem}

\begin{proof}

We first prove that the mechanism is truthful.
For any point $p\in \mathbb R^d$, we can express the squared distance from $p$ to the facility location $\M(\vect x)$ purely as a function of the distances from $p$ to the agent locations, using the equality $\sum_{i=1}^n (\bar x-x_i)=0$:
\begin{align*}
||p-\M(\vect x)||^2 &= ||p-\bar x||^2 + \sigma^2 \\
&=\frac1n \sum_{i=1}^n \left( ||p-\bar x||^2 + ||\bar x-x_i||^2 \right) +  (p-\bar x) \cdot  \vec{0}\\
&=\frac1n \sum_{i=1}^n \left( ||p-\bar x||^2 + ||\bar x-x_i||^2 \right) +  \frac1n\sum_{i=1}^n 2(p-\bar x) \cdot  (\bar x-x_i)\\
&= \frac1n \sum_{i=1}^n ||p-\bar x +\bar x-x_i||^2 = \frac1n \sum_{i=1}^n ||p-x_i||^2.
\end{align*}

Therefore, if the true location of agent $i$ is at $p$, the best reported location $x_i$ for agent $i$ is equal to $p$, hence the truthfulness.

We now prove the approximation factor of \emph{SD-lift}. Let $j$ be the agent furthest from the facility location. We have:
\begin{align*}
||x_j-\M(\vect x)||^2 &= \frac1n \sum_{i=1}^n ||x_j-x_i||^2 \leq \frac1n \cdot (n-1) \cdot (2\OPT)^2.
\end{align*}

Indeed, any two distinct agents are separated by a distance at most $2\OPT$ by definition. Therefore, the approximation factor is at most $2\sqrt{\frac{n-1}{n}}$.

This bound is tight for the mechanism. Consider two points at
distance $2$, with one agent at one point and the remaining $n-1$ agents
at the other. Then $\OPT=1$, while for the isolated agent $x_j$,
\[
\|x_j-\M(\vect{x})\|^2
=
\frac{1}{n}(n-1)\cdot 2^2
=
4\cdot\frac{n-1}{n}.
\]
Hence, the approximation factor is exactly
$2\sqrt{(n-1)/n}$.

\end{proof}

\paragraph{Acknowledgements}

{This work is supported by the ANR project Predictions ANR-23-CE48-0010 and the MIAI Chaire Frugal Artificial Intelligence ANR-23-IACL-0006.}

\bibliographystyle{plainnat}
\bibliography{bibfile}

\appendix

\end{document}